\documentclass[12pt]{article}

\usepackage{bm}
\usepackage{amsmath}
\usepackage{amssymb}
\usepackage{amsthm} 
\usepackage{epsfig}
\usepackage{graphicx}
\usepackage{latexsym} 
\usepackage{multicol}
\usepackage{psfrag}

\newcommand{\lenz}{{\bf L}}
\newcommand{\erre}{{\bf r}}
\newcommand{\erredot}{\dot{\bf r}}
\newcommand{\qu}{{\bf q}}
\newcommand{\qudot}{\dot{\bf q}}
\newcommand{\angmom}{{\bf c}}
\newcommand{\energy}{\mathcal{E}}

\newcommand{\bDelta}{\bm{\Delta}}

\newcommand{\bK}{{\bf K}}

\newcommand{\bxi}{\bm{\xi}}

\newtheorem{lemma}{\bf Lemma}
\newtheorem{proposition}{\bf Proposition}
\newtheorem{remark}{\bf Remark}
\newtheorem{theorem}{\bf Theorem}

\def\R{\mathbb{R}}
\def\C{\mathbb{C}}

\def\DD{{\bf D}}
\def\EE{{\bf E}}
\def\FF{{\bf F}}
\def\GG{{\bf G}}
\def\JJ{{\bf J}}
\def\cc{{\bf c}}
\def\rhodot{\dot{\rho}}
\def\q{{\bf q}}  
\def\dq{\dot{\bf q}}

\def\bv{{\bf v}}

\def\Att{{\cal A}}
\def\alphadot{\dot{\alpha}} 
\def\deltadot{\dot{\delta}}

\def\bzero{{\bf 0}}   

\def\Rvec{{\bf R}}

\def\Avec{{\bf A}}

\def\bPhi{\bm{\Phi}}
\def\bPsi{\bm{\Psi}}

\def\erho{{\bf e}^\rho}
\def\ealpha{{\bf e}^\alpha}
\def\edelta{{\bf e}^\delta}
\def\eort{{\bf e}^\perp}

\def\qrho{\q\cdot\erho}

\def\qort{\q\cdot\eort}

\def\dqrho{\dq\cdot\erho}

\begin{document} 
\title{\bf Orbit Determination with the two-body Integrals. III}

\author{{\bf G.~F. Gronchi}\footnote{{\tt gronchi@dm.unipi.it}}, 
        {\bf  G. Ba\`u}\footnote{{\tt bagiugio@gmail.com}},
        {\bf S. Mar\`o}\footnote{{\tt maro@mail.dm.unipi.it}} \\
     {Dipartimento di Matematica, Universit\`a di Pisa,}\\ 
     {Largo B. Pontecorvo, 5, Pisa, Italy}}

 
\maketitle

\begin{abstract}
  We present the results of our investigation on the use of the
  two-body integrals to compute preliminary orbits by linking too
  short arcs of observations of celestial bodies.  This work
  introduces a significant improvement with respect to the previous
  papers on the same subject \cite{gdm10}, \cite{gfd11}.  Here we find
  a univariate polynomial equation of degree 9 in the radial distance
  $\rho$ of the orbit at the mean epoch of one of the two arcs.  This
  is obtained by a combination of the algebraic integrals of the
  two-body problem. Moreover, the elimination step, which in
  \cite{gdm10}, \cite{gfd11} was done by resultant theory coupled with
  the discrete Fourier transform, is here obtained by elementary
  calculations.  We also show some numerical tests to illustrate the
  performance of the new algorithm.
\end{abstract}


\section{Introduction}

This paper is related to the research started in \cite{gdm10},
\cite{gfd11}, where some first integrals of the two-body problem were
used to write polynomial equations for the linkage problem with
optical observations of asteroids and space debris.

\noindent In \cite{gdm10} the authors consider polynomial equations
with total degree 48, that are consequences of the conservation of
angular momentum and energy, and they propose two methods to search
for all the solutions using algebraic elimination theory.
%
The same equations were first introduced in \cite{th77}, but their algebraic
character was not fully exploited at that time.
In \cite{gfd11} the authors introduce for the same purpose 
new polynomial equations, with total degree 20, using the angular momentum
integral and a projection of Laplace-Lenz vector along a suitable
direction.
Both in \cite{gdm10} and in \cite{gfd11} the number of the considered
equations is equal to the number of unknowns.

\noindent Here we improve significantly the previous results by writing an
overdetermined polynomial system (more equations than unknowns) which is
proved to be generically consistent, i.e. the set of solutions in the complex
field is not empty. For generic values of the data, by variable elimination,
we obtain a system of two univariate polynomials of degree 10 with a greatest
common divisor of degree 9.

\noindent We also discuss the case where the oblateness of the Earth
is relevant, so that we add the perturbation of the $J_2$
term to the Keplerian potential. In \cite{ftmr10} and \cite{gfd11}
this problem was faced by an iterative scheme, writing at each step
polynomial equations with the same algebraic structure as in the
unperturbed case (without $J_2$ effect).
In this paper, the overdetermined system that can be written following the
steps of the unperturbed case has the same algebraic structure but is
generically inconsistent.  However, we can use the iterative scheme
mentioned above by neglecting one polynomial of the system.

This paper is organized as follows.  After recalling the linkage
problem and some preliminaries on the two-body integrals
(Sections~\ref{s:linkage}, \ref{s:integrals}), in
Section~\ref{s:polylink} we discuss the polynomial equations that can
be written for this problem, including the overdetermined
system (\ref{sys2solve}) which is the object of this work.
In Section~\ref{s:elim} we show how the elimination steps can be
carried out, and we prove the consistency of system (\ref{sys2solve}).
In Section~\ref{s:selsol} we discuss the spurious solutions of
(\ref{sys2solve}), and illustrate two methods to discard them.
Section~\ref{s:J2} is devoted to the linkage problem with the $J_2$
effect.
Some numerical tests are presented in Section~\ref{s:numexp}.
In the Appendix we discuss a simple way to filter the pairs of
attributables to be linked with this algorithm.

\section{Linkage of too short arcs}
\label{s:linkage}

We consider objects moving in a central force field. 
Let us fix an inertial reference frame, with the origin at the center of
attraction $O$, which is the center of the Sun (Earth) in the asteroid
(space debris) case.
Assume the position $\q$ and velocity $\dq$ of the observer are
known functions of time.  We describe the position of the observed body as the
sum
\[
\erre = \q + \rho\erho,
\]
with $\rho$ the topocentric distance and $\erho$ the {\em line of sight} unit
vector.
We choose spherical coordinates $(\alpha, \delta, \rho)\in [-\pi,\pi)\times
(-\pi/2,\pi/2)\times \R^+$, so that
\[
\erho = (\cos\delta\cos\alpha,\cos\delta\sin\alpha,\sin\delta) .
\]
A typical choice for $\alpha,\delta$ is right ascension and declination.
The velocity vector is
\[
\erredot = \dq + \dot{\rho}\erho +
\rho(\alphadot\cos\delta\ealpha + \deltadot\edelta),
\qquad
\rhodot,\alphadot,\deltadot\in\R, \rho\in\R^+,
\]
where $\rhodot$, $\rho\alphadot\cos\delta$, $\rho\deltadot$ are the components
of the velocity, relative to the observer,
in the (positively oriented)
orthonormal basis
$\{\erho,\ealpha,\edelta\}$, with
\[
\ealpha =
(\cos\delta)^{-1}\frac{\partial \erho}{\partial\alpha},\hskip 1cm \edelta =
\frac{\partial \erho}{\partial\delta} .
\]

Let $(t_i, \alpha_i, \delta_i)$ with $i=1\ldots m$, $m\geq 2$, be a
short arc of optical observations of a moving body, made from the same
station.  
%
If $m\geq 3$, we can compute $\alpha$, $\delta$, $\alphadot$, $\deltadot$,
$\ddot\alpha$, $\ddot\delta$ at the mean time $\bar t =
\frac{1}{m}\sum_{i=1}^mt_i$ by a quadratic fit. From these quantities we can
try to compute a preliminary orbit.  When the second derivatives are not
reliable due to errors in the observations (or not available, if $m=2$) we
speak of a too short arc (TSA) and, to compute a preliminary orbit, we
have to add information coming from other arcs of observations. This is a
typical identification problem, see \cite{mg2010}.


\noindent In any case, it is possible to compute an {\em attributable}
\[
\Att = (\alpha, \delta,
\alphadot, \deltadot) \in [-\pi,\pi) \times
(-\pi/2,\pi/2) \times \R^2,
\]
representing the angular position and velocity of the body at epoch
$\bar t$ (see \cite{ident4}, \cite{gdm10}).  The radial
distance and velocity $\rho, \rhodot$ are completely undetermined and are
the missing quantities to define an orbit for the body.

\noindent In this paper we deal with the {\em linkage problem}, that
is to join together two TSAs of observations to form an orbit fitting
all the data.


\section{First integrals of Kepler's motion}
\label{s:integrals}

We consider the first integrals of the equation of Kepler's problem
\[
\ddot\erre = -\frac{\mu}{|\erre|^3}\erre
\]
as functions of the unknowns $\rho$, $\rhodot$.
The angular momentum is 
the polynomial vector
\[
\cc(\rho,\rhodot) = \erre \times \erredot = \DD \rhodot + \EE \rho^2 + \FF \rho
+ \GG,
\]
with
\[
\DD = \q\times\erho,\quad\
%
\EE = \erho\times\eort,\quad\
%
\FF = \q\times\eort + \erho\times\dq,\quad\
%
\GG = \q\times\dq,
\]
where we have set
\[
\eort = \alphadot\cos\delta \ealpha + \deltadot\edelta.
\]
\noindent The expression of the energy is
\begin{equation}
{\cal E}(\rho,\rhodot) = \frac{1}{2}|\erredot|^2 - \frac{\mu}{|\erre|},
\label{energy}
\end{equation}
where 
\begin{eqnarray}
|\erre| &=& (\rho^2 + |\q|^2 + 2\rho\qrho)^{1/2},\label{erre} \\
\vert\erredot\vert^2 &=& \rhodot^2 
+ \vert\eort\vert^2\rho^2 + 2\dqrho\rhodot + 
2\dq\cdot\eort\rho + \vert\dq\vert^2. \label{erredot2}
\end{eqnarray}
The Laplace-Lenz vector $\lenz$ is given by
\begin{equation}
\mu\lenz(\rho,\rhodot) = \erredot\times\angmom -
\mu\frac{\erre}{|\erre|} = \Bigl(\vert\erredot\vert^2
  -\frac{\mu}{|\erre|}\Bigr)\erre - (\erredot\cdot\erre)\erredot,
\label{laplace_lenz}
\end{equation}
with $|\erre|$, $\vert\erredot\vert^2$ as in (\ref{erre}), (\ref{erredot2}),
and
\begin{eqnarray*}
\erredot\cdot\erre &=& \rho\rhodot + \qrho\rhodot + 
(\dqrho + \qort)\rho +  \dq\cdot\q .
\end{eqnarray*}
Moreover, the following relations hold for all $\rho,\rhodot$:
\begin{equation}
\angmom\cdot\lenz = 0,
\hskip 1cm
\mu^2|\lenz|^2 = \mu^2 + 2\energy|\angmom|^2.
\label{intrel}
\end{equation}
Expressions (\ref{energy}), (\ref{laplace_lenz}) are algebraic, but not
polynomial, in $\rho,\rhodot$. However, we can introduce a new variable
$u\in\R$, together with the relation $|\erre|^2u^2=\mu^2$ and we obtain
\[
\energy = \frac{1}{2}|\erredot|^2 - u,
\hskip 1cm
\mu\lenz = (|\erredot|^2 - u)\erre - (\erredot\cdot\erre)\erredot,
\]
that are polynomials in $\rho, \rhodot, u$.
%

\noindent For later reference we also introduce the quantity
\begin{equation}
  \bK = \frac{1}{2}\vert\erredot\vert^2\erre -
  (\erredot\cdot\erre)\erredot
\label{kappa} .
\end{equation}

\section{Polynomial equations for the linkage}
\label{s:polylink}

We use the notation above, with index 1 or 2 referring to the epoch.
Let 
\[
\Att_j =(\alpha_j,\delta_j,\alphadot_j,\deltadot_j), \qquad j=1,2 
\]
be two attributables at epochs $\bar{t}_j$.
We consider the polynomial system
\begin{equation}
\angmom_1 =\angmom_2, \quad
\lenz_1 = \lenz_2, \quad
\energy_1 = \energy_2, \quad
u_1^2|\erre_1|^2 = \mu^2, \quad
u_2^2|\erre_2|^2 = \mu^2,
\label{fullsys}
\end{equation}
in the 6 unknowns
\[
(\rho_1,\rho_2,\rhodot_1,\rhodot_2,u_1,u_2).
\]
System (\ref{fullsys}) is defined by the vector of parameters
\[
(\Att_1,\Att_2,\qu_1,\qu_2,\qudot_1,\qudot_2),
\]
and we shall discuss properties which hold for generic values of
them.
Moreover, (\ref{fullsys}) is composed by 9 equations with 6 unknowns.
However, due to relations (\ref{intrel}), 2 equations can be
considered as consequences of the others, so that we are left with a
system of 7 equations with 6 unknowns.


Assume the attributables $\Att_1,\Att_2$ refer to the same
observed body, the two-body dynamics is perfectly respected, and there
are no observing errors.  Then the set of solutions of (\ref{fullsys})
is not empty.
Taking into account the
observational errors, and the fact that the two-body motion is only an
approximation, system (\ref{fullsys}) turns out to be generically
inconsistent, i.e. it has no solution in $\C$.

We search for a polynomial system, consequence of (\ref{fullsys}), which is
generically consistent, with a finite number of solutions in $\C$, and which
leads by elimination to a univariate polynomial equation of the lowest degree
possible.

Introducing relations $u_j^2|\erre_j|^2=\mu^2$ ($j=1,2$) for the
auxiliary variables $u_1$, $u_2$ corresponds to the squaring
operations, used in \cite{gdm10}, \cite{gfd11} to bring the selected
algebraic system in the variables
$(\rho_1,\rho_2,\rhodot_1,\rhodot_2)$ into a polynomial form.  Since
these operations are responsible of the high total degree of the
resulting polynomial systems (48 and 20 respectively), in writing the
new equations we try to cancel the dependence on both $u_1$, $u_2$ by
algebraic manipulations of the conservation laws.
First we shall consider the intermediate system
\begin{equation}
\angmom_1=\angmom_2,\quad
\mu(\lenz_1 -\lenz_2) = (\energy_1-\energy_2)\erre_2,\quad
u_1^2|\erre_1|^2 = \mu^2,
\label{intersys}
\end{equation}
where $u_2$ does not appear, which is still inconsistent; 
then we shall take into account the system
\begin{equation}
\angmom_1=\angmom_2,\quad
(\bK_1 - \bK_2)\times(\erre_1-\erre_2) = \bzero,
\label{sys2solve}
\end{equation}
where also $u_1$ does not appear, whose consistency is proven in the next
section.

\section{Elimination of variables}
\label{s:elim}

In this section we show that generically system (\ref{sys2solve}) is
consistent.  In particular, by elimination of variables, we shall end up with
two univariate polynomials of degree 10 in the range $\rho_2$, whose greatest
common divisor generically has degree 9.  A similar procedure can be carried
out by eliminating all the variables but $\rho_1$.

\subsection{Angular momentum equations}
\label{s:angmom}

The conservation of angular momentum gives us 3 polynomial equations
that are linear in $\rhodot_1, \rhodot_2$, and quadratic in
$\rho_1,\rho_2$. Therefore, it is natural to use these equations to
eliminate the radial velocities, as done in \cite{gdm10},
\cite{gfd11}.
These equations can be written as
\begin{equation}
\DD_1\rhodot_1 - \DD_2\rhodot_2 = \JJ(\rho_1,\rho_2),
\label{eq_AM}
\end{equation}
with $\JJ$ a vector whose components are quadratic polynomials in $\rho_1$,
$\rho_2$.
Following \cite{gdm10} we project (\ref{eq_AM}) onto the
vectors $\DD_2\times(\DD_1\times\DD_2)$ and
$\DD_1\times(\DD_1\times\DD_2)$ and obtain $\rhodot_1$, $\rhodot_2$ as
quadratic polynomials in $\rho_1$, $\rho_2$.  With these expressions
of $\rhodot_1$, $\rhodot_2$ we have
\begin{equation}
\bDelta_c\times(\DD_1\times\DD_2)=\bzero,
\label{direction_Deltac}
\end{equation}
with $\bDelta_c= \angmom_1-\angmom_2$, whatever the values of $\rho_1,
\rho_2$.
\noindent The projection of (\ref{eq_AM}) onto $\DD_1\times\DD_2$ allows
us to eliminate the variables $\rhodot_1, \rhodot_2$ and yields
\begin{equation}
q(\rho_1,\rho_2) = q_{2,0}\rho_1^2 + q_{1,0}\rho_1 +
q_{0,2}\rho_2^2 + q_{0,1}\rho_2 + q_{0,0},
\label{quad_form}
\end{equation}
where the coefficients $q_{i,j}$ depend only on the attributables and
on the position and velocity of the observer at epochs $\bar{t}_1$,
$\bar{t}_2$.


\noindent In the following we shall consider the quantities introduced in
Section~\ref{s:integrals} as function of $\rho_1$, $\rho_2$ only, by
the elimination of $\rhodot_1$, $\rhodot_2$ just recalled.

\subsection{Bivariate equations for the linkage}

By subtracting $\energy_2\erre_2$ to both members of the Laplace-Lenz
equation, and using the conservation of energy, we obtain
\begin{equation}
\mu\lenz_1 -\energy_1\erre_2=\mu\lenz_2-\energy_2\erre_2.
\label{u2elim}
\end{equation}
Equation (\ref{u2elim}) can be written
\begin{equation}
\bDelta_K + (\frac{1}{2}|\erredot_1|^2 - u_1)\bDelta_r= \bzero,
\label{mixLE}
\end{equation}
where we have set
\[
\bDelta_K = \bK_1 - \bK_2,
\hskip 1cm
\bDelta_r = \erre_1 - \erre_2,
\]
with $\bK$ as in (\ref{kappa}).
Note that the variable $u_2$ does not appear in (\ref{mixLE}).


\noindent We can also eliminate $u_1$ by cross product with $\bDelta_r$:
\begin{equation}
\bDelta_K\times \bDelta_r = \bzero.
\label{eqdeg6}
\end{equation}
For brevity we set $\bxi = \bDelta_K\times \bDelta_r$,
and we note that
\begin{equation}
\bxi =
\frac{1}{2}(|\erredot_2|^2- |\erredot_1|^2)\erre_1\times\erre_2
- (\erredot_1\cdot\erre_1)\erredot_1\times\bDelta_r +
(\erredot_2\cdot\erre_2)\erredot_2\times\bDelta_r.
\label{DeltaKxDeltar}
\end{equation}

\begin{remark}
By developing the expressions of $\erre_1\times\erre_2$,
$\erredot_1\times\bDelta_r$, $\erredot_2\times\bDelta_r$ as polynomials in 
$\rho_1,\rho_2$ we obtain that the monomials in $\bxi$
with the highest total degree, which is 6, are all multiplied by
$\erho_1\times\erho_2$.
\label{rem:deg6}
\end{remark}

\noindent In the following section we shall prove that for generic values of
the data the system
\begin{equation}
q = 0,
\qquad
\bxi= \bzero
\label{qxi}
\end{equation}
is consistent, that is the set of its roots in $\C$ is not empty.

\noindent Note that, if $q=0$, the vector $\bxi$ is parallel to the common
value $\angmom_1=\angmom_2$ of the angular momentum.

\subsection{Consistency of the equations}

The proof of the consistency relies on some geometrical considerations. In
particular it is relevant to check whether the angular momentum vector is
orthogonal to the line of sight.
We introduce the quantities
\[
c_{ij} = \angmom_i\cdot\erho_j, \hskip 1cm i,j=1,2.
\]
More explicitly we have
\begin{eqnarray*}
c_{11} &=&  \qu_1\times\erho_1\cdot\eort_1 \rho_1 +
\qu_1\times\erho_1\cdot\qudot_1,
\\
c_{12} &=& \erho_1\times\eort_1\cdot\erho_2\rho_1^2 +
\qu_1\times\erho_1\cdot\erho_2\rhodot_1(\rho_1,\rho_2) +\\ 
&+& (\erho_1\times\qudot_1 + \qu_1\times\eort_1)\cdot\erho_2\rho_1 +
\qu_1\times\qudot_1\cdot\erho_2,
%
\end{eqnarray*}
and similar expressions for $c_{22}$, $c_{21}$.  
In particular, equations $c_{11}=0$ and $c_{22}=0$ represent
straight lines in the plane $\rho_1\rho_2$, while $c_{12}=0$ and
$c_{21}=0$ give conic sections, see Figure~\ref{f:orthrel}.

\begin{figure}[!h]
\centerline{\includegraphics[width=10cm]{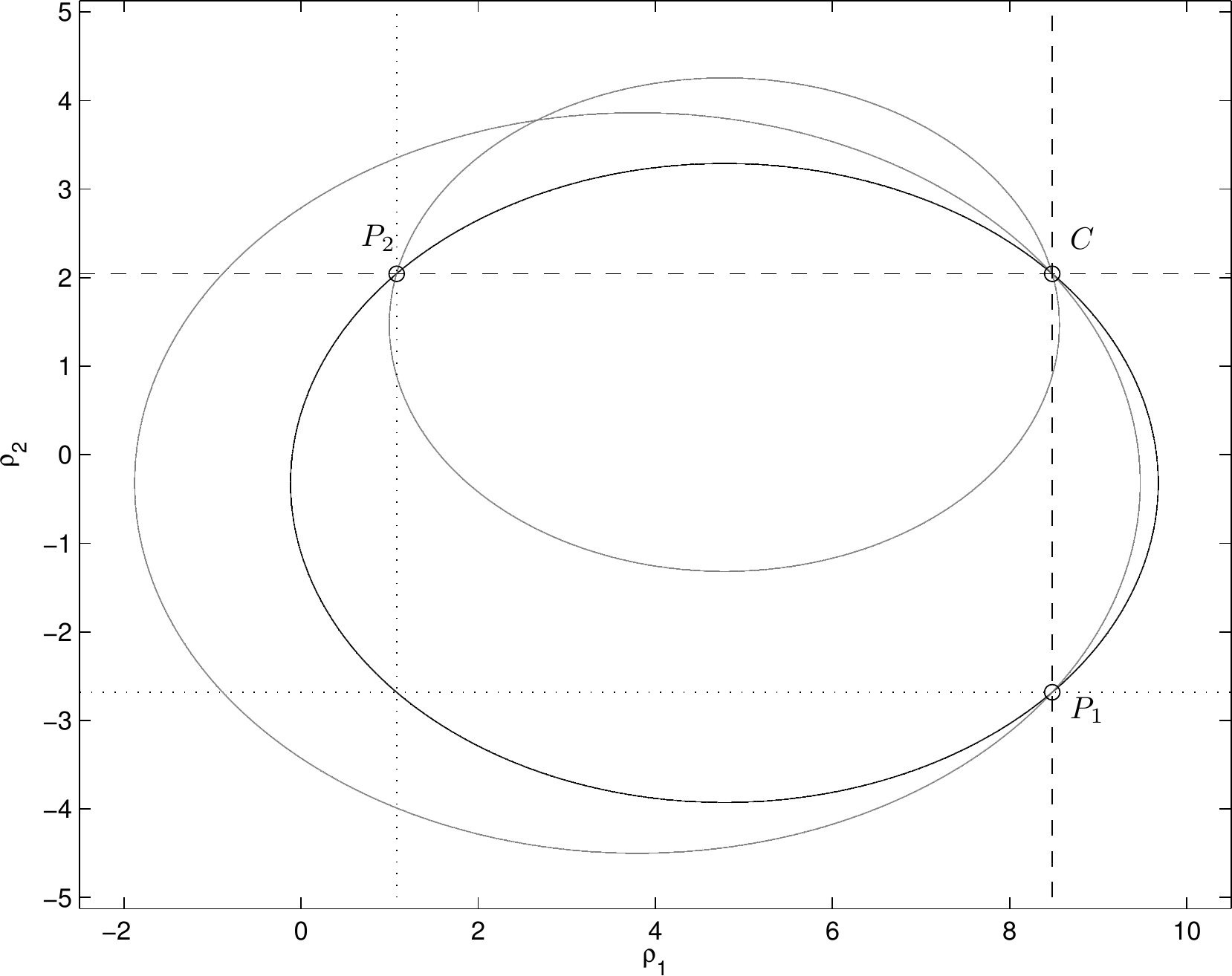}}
\caption{In the plane $\rho_1\rho_2$, for a test case, we draw the
  curves $q=0$ (black), $c_{12}=0$ and $c_{21}=0$ (light gray), and the straight lines $c_{11}=0$, $c_{22}=0$ (dashed),
and $\rho_1=\rho_1', \rho_2=\rho_2'$ (dotted).}
\label{f:orthrel}
\end{figure}

\noindent Consider the point $C=(\rho_1'',\rho_2'')$ defined by $c_{11}=c_{22}
= 0$, so that
\begin{equation}
\rho_1'' =
\frac{\qu_1\times\qudot_1\cdot\erho_1}{\erho_1\times\eort_1\cdot\qu_1},
\qquad \rho_2'' =
\frac{\qu_2\times\qudot_2\cdot\erho_2}{\erho_2\times\eort_2\cdot\qu_2}.
\label{rhojsecondo}
\end{equation}

\noindent In Lemma~\ref{conicnotempty} we shall prove that $C$ lies on the
conic $q=0$ and is the only point where both angular momenta $\angmom_1$,
$\angmom_2$ vanish.  The straight line $\rho_2 = \rho_2''$ generically meets
$q=0$ in another point $P_2=(\rho_1',\rho_2'')$, where the angular momenta do
not vanish.  Similarly, the straight line $\rho_1 = \rho_1''$ generically
meets $q=0$ in another point $P_1=(\rho_1'',\rho_2')$, where the angular
momenta are not zero, see Figure~\ref{f:orthrel}.  For $q=0$, the vector
$\erre_1\times\erre_2$ gives the direction of $\angmom_1=\angmom_2$, therefore
from the equations $\erre_1\times\erre_2\cdot\erho_j=0$, $j=1,2$ we obtain
\begin{equation}
\rho_1' = \frac{\qu_1\times\qu_2\cdot\erho_2}{\erho_1\times\erho_2\cdot\qu_2},
\qquad
\rho_2' = \frac{\qu_1\times\qu_2\cdot\erho_1}{\erho_1\times\erho_2\cdot\qu_1}.
\label{rhojprimo}
\end{equation}

\noindent Note that 
\begin{eqnarray*}
q_{2,0} &=& -\EE_1\cdot\DD_1\times\DD_2 = 
-(\erho_1\times\eort_1\cdot\qu_1)(\erho_1\times\erho_2\cdot\qu_2),
\end{eqnarray*}
so that $q_{2,0}\neq 0$ implies that both $\rho_1'$ and $\rho_1''$ are well
defined.

\noindent In a similar way we obtain
\begin{eqnarray*}
q_{0,2} &=& \EE_2\cdot\DD_1\times\DD_2 = 
(\erho_2\times\eort_2\cdot\qu_2)(\erho_1\times\erho_2\cdot\qu_1),
\end{eqnarray*}
so that $q_{0,2}\neq 0$ implies that both $\rho_2'$ and $\rho_2''$ are well
defined.

\noindent Generically we have
\begin{equation}
q_{2,0}, q_{0,2}\neq 0.
\label{conicnondeg}
\end{equation}

\begin{lemma}
  If (\ref{conicnondeg}) holds, then the point
  $C=(\rho_1'',\rho_2'')$ given by $c_{11}=c_{22} = 0$ satisfies
  $q(\rho_1'',\rho_2'') = 0$.  Moreover, in $C$ we have $\angmom_1 =
  \angmom_2=\bzero$ and $C$ is the unique point in the plane $\rho_1\rho_2$
  where both angular momenta vanish.
\label{conicnotempty}
\end{lemma}

\begin{proof}
Using relation $\erre_j\cdot\DD_j=0$ we obtain
\[
\angmom_j\times\DD_j = -(\erredot_j\cdot\DD_j)\erre_j,\hskip 1cm j=1,2.
\]
Moreover, condition (\ref{conicnondeg}) yields $\erho_j\times\q_j\neq\bzero$,
so that $\erre_j\neq\bzero$.

\noindent From relations
\[
\erredot_j\cdot\DD_j = -\angmom_j\cdot\erho_j 
\]
we have
\begin{equation}
\angmom_j\times\DD_j= \bzero \ \mbox{ if and only if }\  c_{jj}=0
\label{angmomparallD}
\end{equation}
for $j=1,2$.
Therefore $c_{11} = c_{22} = 0$
implies
\[
\bDelta_c\cdot\DD_1\times\DD_2 = 0,
\]
that together with (\ref{direction_Deltac}) gives
\begin{equation}
\angmom_1=\angmom_2.
\label{equalc}
\end{equation}
Finally, relations (\ref{angmomparallD}), (\ref{equalc})  imply
$\angmom_j=\bzero$, $j=1,2$. The uniqueness immediately follows from the
definition of $C$.

\rightline{$\square$\hskip 0.5cm}

\end{proof}

\begin{lemma}
In the point $C=(\rho_1'',\rho_2'')$ generically we have $\bxi\neq\bzero$.
\label{czeroxinonzero}
\end{lemma}
\begin{proof}
  By Lemma~\ref{conicnotempty}, if (\ref{conicnondeg}) holds, we have
  $\angmom_1=\angmom_2=\bzero$ in $C$, so that
\begin{eqnarray*}
&&\mu(\lenz_1-\lenz_2) - (\energy_1-\energy_2)\erre_2 = 
\mu\Bigl(\frac{\erre_2}{|\erre_2|} - \frac{\erre_1}{|\erre_1|}\Bigr) -
(\energy_1-\energy_2)\erre_2 \\
&=& -\mu\frac{\bDelta_r}{|\erre_1|} +
\frac{1}{2}(|\erredot_2|^2 - |\erredot_1|^2)\erre_2.
\end{eqnarray*}
Therefore 
we have
\begin{equation}
\bDelta_K\times\bDelta_r = \frac{1}{2}(|\erredot_1|^2 -
|\erredot_2|^2)\erre_1\times\erre_2.
\label{dKdrceq0}
\end{equation}
We show that the right-hand side of (\ref{dKdrceq0}) does
not vanish in $C$.  
In fact, by projecting $\erre_1\times\erre_2$ onto $\qu_1, \qu_2$
we obtain
\begin{eqnarray*}
&&\qu_1\cdot\erre_1\times\erre_2 =
  \rho_1''\left[\rho_2''(\erho_1\times\erho_2\cdot\qu_1) -
    \qu_1\times\qu_2\cdot\erho_1 \right],\\
&&\qu_2\cdot\erre_1\times\erre_2 =
  \rho_2''\left[\rho_1''(\erho_1\times\erho_2\cdot\qu_2) -
    \qu_1\times\qu_2\cdot\erho_2 \right],
\end{eqnarray*}
and the expressions in the brackets vanish only if
\[
\rho_1'=\rho_1'',\qquad
\rho_2'=\rho_2''.
\]
Moreover, $\rho_1''=\rho_2''=0$ occurs only if
$\qu_j\times\qudot_j\cdot\erho_j= 0$, $j=1,2$. Thus
$\erre_1\times\erre_2$ generically does not vanish.
Using $\angmom_1=0$ projected e.g. onto $\eort_1$, we can prove that
in $C$ the quantity $|\erredot_1|^2$ depends only on the data $\Att_1,
\qu_1, \qudot_1$ at epoch $\bar{t}_1$. In fact through this relation
we can find an expression for $\rhodot_1$ depending only on $\Att_1, \qu_1,
\qudot_1$. A similar result holds for $|\erredot_2|^2$.
We observe that the $|\erredot_j|^2$ do not vanish
individually.  
Indeed, if relations
$\erho_j\times\eort_j\cdot\qudot_j\neq 0$, $j=1,2$
hold, then $\erredot_j = \rhodot_j\erho_j + \rho_j\eort_j +
\qudot_j\neq\bzero$ for any choice of $\rho_j,\rhodot_j$.
Therefore generically also $|\erredot_1|^2 - |\erredot_2|^2$ does not vanish.
We conclude that the point $C$ is not a solution of
(\ref{qxi}).

\rightline{$\square$\hskip 0.5cm}

\end{proof}

\begin{lemma}
Assume $q=0$. Then $\bxi=\bzero$ is generically equivalent to
\begin{equation}
\left\{
\begin{array}{l}
\bxi\cdot\erho_1 =0\cr
\angmom\cdot\erho_1 \neq 0\cr
\end{array}
\right.
\quad \mbox{ or }\quad
\left\{
\begin{array}{l}
\bxi\cdot\erho_2 = 0\cr
\angmom\cdot\erho_2 \neq 0\cr
\end{array}
\right.
.
\label{alternative}
\end{equation}
\label{lem:altern}
\end{lemma}
\begin{proof}
  Assume (\ref{alternative}) does not hold.  Clearly relations
  $\bxi\cdot\erho_1 = \bxi\cdot\erho_2 =0$ are necessary to have
  $\bxi=\bzero$.  Then we have $\angmom\cdot\erho_1=\angmom\cdot\erho_2=0$.
  If (\ref{conicnondeg}) holds, Lemma~\ref{conicnotempty} implies
  $\angmom_1=\angmom_2=\bzero$ and, by Lemma~\ref{czeroxinonzero}, generically
  we have $\bxi\neq\bzero$.
  Viceversa, assuming $\bxi\neq\bzero$ we obtain that
  each system in (\ref{alternative}) is incompatible, because
  for $q=0$ we have $\bxi\times\angmom=\bzero$.

\rightline{$\square$\hskip 0.5cm}

\end{proof}

\begin{lemma}
Assume relation
\begin{equation}
\bDelta_q\cdot\erho_1\times\erho_2\neq 0
\label{basis}
\end{equation}
holds. Then
 $\bxi=\bzero$ is equivalent to
\[
\bxi\cdot\erho_1=0 \qquad \mbox{ and }\qquad \bxi\cdot\erho_2=0.
\]
\label{lem:basis}
\end{lemma}
\begin{proof}
Relation $\bxi\cdot\bDelta_r=0$ holds trivially. Moreover, since
(\ref{basis}) is satisfied, the vectors $\erho_1,\erho_2,\bDelta_r$ are
linearly independent.

\rightline{$\square$\hskip 0.5cm}

\end{proof}

\noindent We introduce the polynomials
\[
p_1 = \bxi\cdot\erho_1,
\hskip 0.7cm
p_2 = \bxi\cdot\erho_2.
\]
By Remark~\ref{rem:deg6} both $p_1$ and $p_2$ have total degree 5 in the
variables $\rho_1$, $\rho_2$. 

\noindent We consider the system
\begin{equation}
 q = p_1 = p_2 = 0.
\label{p1p2q}
\end{equation}



We are now ready to state the main result.
\begin{theorem}
  Generically, system (\ref{qxi}) is consistent and can be reduced to a system
  of two univariate polynomials $\mathfrak{u}_1,\mathfrak{u}_2$ whose greatest
  common divisor has degree 9.
\end{theorem}
\begin{proof}
  By Lemma~\ref{lem:basis} we only need to prove consistency of system
  (\ref{p1p2q}).
  First we show that system (\ref{qxi}) has at least $9$ solutions. In fact
  Lemma~\ref{lem:altern} implies that both systems $q=p_1=0$ and $q=p_2=0$
  define generically 10 points. Moreover, for $q=0$ relation $c_{11}\neq 0$
  discards the points $P_1$, $C$, while relation $c_{22}\neq 0$ discards the
  points $P_2$, $C$.  By Lemma~\ref{czeroxinonzero} we know that generically
  $C$ is not a solution of $\bxi=\bzero$, hence it does not solve either
  $p_1=0$ or $p_2=0$.  On the other hand, we have $p_1(P_1) = p_2(P_2)= 0$.
  We show that $p_1(P_1) = 0$.  If $\bxi(P_1)=\bzero$ the results
  trivially holds.  If $\bxi(P_1)\neq\bzero$, by Lemma~\ref{conicnotempty} we
  have $\angmom(P_1)\neq\bzero$, so that $\angmom(P_1)\parallel \bxi(P_1)$,
  and from $c_{11}(P_1)=0$ we obtain $p_1(P_1)=0$.
  In a similar way we can prove that $p_2(P_2) = 0$.  Therefore we are left
  with 9 solutions for each system. We show that they are the same ones.  In
  fact, by Lemma~\ref{lem:basis}, solutions of (\ref{qxi}) must satisfy
  $p_1=p_2=0$. Moreover, generically we have $p_1(P_2) \neq 0$ and $p_2(P_1)
  \neq 0$, so that both $P_1$ and $P_2$ are not solutions.  Therefore we have
  exactly 9 solutions.

  \noindent Let us consider the univariate polynomials
  \begin{equation}
    \mathfrak{u}_j= \mathrm{res}(p_j,q,\rho_1),
    \qquad
    j=1,2,
    \label{uj}
  \end{equation}
  that are the resultant of the pairs $p_j,q$ with respect to $\rho_1$ (see
  \cite{cox}).

  \noindent The root $\rho_2=\rho_2'$ of $\mathfrak{u}_1$ and the root
  $\rho_2=\rho_2''$ of $\mathfrak{u}_2$ must be discarded because they
  correspond to the points $P_1$, $P_2$ for the polynomials $p_1, p_2$
  respectively.

  \noindent We consider
  \[
  \tilde{\mathfrak{u}}_1 = \frac{\mathfrak{u}_1}{\rho_2-\rho_2'}, \qquad \tilde{\mathfrak{u}}_2 =
  \frac{\mathfrak{u}_2}{\rho_2-\rho_2''}.
  \]
  By the discussion above, these polynomials have degree 9 and must have the
  same roots: in particular, up to a constant factor, they both correspond to
  the greatest common divisor of $\mathfrak{u}_1$ and $\mathfrak{u}_2$.

  \rightline{$\square$\hskip 0.5cm}

\end{proof}

\subsection{The univariate polynomials $\mathfrak{u}_1$, $\mathfrak{u}_2$}

We explicitly perform the elimination step to pass from system
$q = p_1 = p_2 = 0$
to
\[
\mathfrak{u}_1 = \mathfrak{u}_2 = 0,
\]
where $\mathfrak{u}_1, \mathfrak{u}_2$ are the two univariate polynomials
introduced in (\ref{uj}).
For this purpose we produce an equivalent system $q = \tilde{p}_1 =
\tilde{p}_2 = 0$ where the $\tilde{p}_j$ are linear in one variable, say
$\rho_1$.

Assume $q$ is not degenerate, i.e. 
$q_{2,0}, q_{0,2}\neq 0$. 
Indeed, here we use $q_{2,0}\neq 0$ only, while $q_{0,2}\neq 0$ is necessary
for the similar construction relative to $\rho_2$.

\noindent We write
\[
q(\rho_1,\rho_2) = \sum_{i,j=0}^2q_{i,j}\rho_1^i\rho_2^j = 
\sum_{h=0}^2b_h(\rho_2)\rho_1^h,
\]
where
\[
b_0(\rho_2) = q_{0,2}\rho_2^2 +  q_{0,1}\rho_2 +  q_{0,0},
\qquad
b_1 = q_{1,0},
\qquad
b_2 = q_{2,0}.
\]
Moreover we write
\begin{eqnarray}
p_1(\rho_1,\rho_2) &=& 
 \sum_{i,j=0}^5p^{(1)}_{i,j}\rho_1^i\rho_2^j = \sum_{h=0}^4a_{1,h}(\rho_2)\rho_1^h,
\label{eqp1}\\
p_2(\rho_1,\rho_2) &=&
 \sum_{i,j=0}^5p^{(2)}_{i,j}\rho_1^i\rho_2^j = 
\sum_{h=0}^5a_{2,h}(\rho_2)\rho_1^h,\label{eqp2}
\end{eqnarray}
for some polynomials $a_{k,h}$ whose degrees are described by the small
circles used to construct Newton's polygons of $p_1,p_2$ in
Figure~{\ref{polynewt}.
From $q=0$
we obtain
\begin{equation}
\rho_1^h = 
\beta_h\rho_1 + \gamma_h, \qquad h=2,3,4,5
\label{rho1h}
\end{equation}
where
\[
\beta_2 = -\frac{b_1}{b_2}, \qquad \gamma_2 = -\frac{b_0}{b_2},
\]
and
\[
\beta_{h+1} = \beta_h\beta_2 + \gamma_h, \qquad \gamma_{h+1} = \beta_h\gamma_2, \qquad h=2,3,4. 
\]
Inserting (\ref{rho1h}) into (\ref{eqp1}), (\ref{eqp2}) we obtain
\begin{equation}
\tilde{p}_j(\rho_1,\rho_2) = \tilde{a}_{j,1}(\rho_2)\rho_1 + 
\tilde{a}_{j,0}(\rho_2), \hskip 1cm  j=1,2
\label{pjtilde}
\end{equation}
where
\begin{eqnarray*}
&&\tilde{a}_{1,1} = a_{1,1} + \sum_{h=2}^4 a_{1,h}\beta_h,
\hskip 0.5cm 
\tilde{a}_{1,0} = a_{1,0} + \sum_{h=2}^4 a_{1,h}\gamma_h,\\
&&\tilde{a}_{2,1} = a_{2,1} + \sum_{h=2}^5 a_{2,h}\beta_h,
\hskip 0.5cm 
\tilde{a}_{2,0} =  a_{2,0} + \sum_{h=2}^5 a_{2,h}\gamma_h.
\end{eqnarray*}

\begin{figure}[!h]
\psfragscanon
\centerline{\includegraphics[width=12cm]{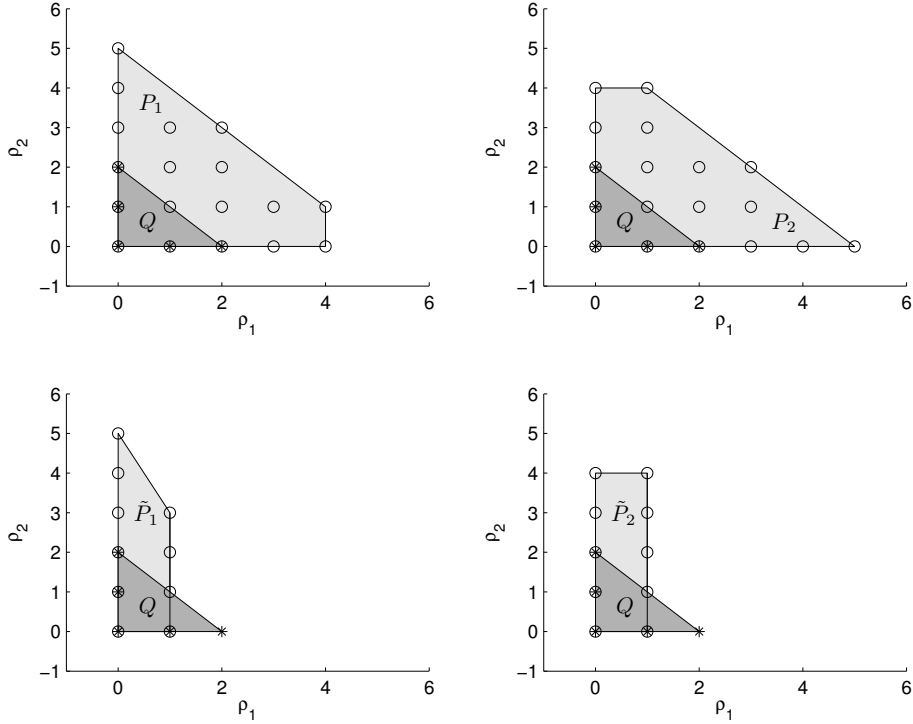}}
\psfrag{Q}{$Q$}\psfrag{P1}{$P_1$}\psfrag{P2}{$P_2$}
\psfrag{TP1}{$\tilde{P}_1$}\psfrag{TP2}{$\tilde{P}_2$}
\psfragscanoff
\caption{Newton's polygons $Q$, $P_j$, $\tilde{P}_j$ for the polynomials $q$,
  $p_j, \tilde{p}_j$, $j=1,2$. The nodes with circles correspond to
  (multi-index) exponents of the monomials in $p_j$, $\tilde{p}_j$; the nodes
  with asterisks correspond to exponents of the monomials in $q$.}
\label{polynewt}
\end{figure}

\noindent In Figure~\ref{polynewt}, bottom left and right, we draw
Newton's polygons of $\tilde{p}_1$, $\tilde{p}_2$, which also
describe the degrees of the polynomials $\tilde{a}_{k,h}$.

\noindent From (\ref{pjtilde}) we get two expressions for $\rho_1$:
\[
\rho_1 = -\frac{\tilde{a}_{1,0}}{\tilde{a}_{1,1}},
\hskip 1cm
\rho_1 = -\frac{\tilde{a}_{2,0}}{\tilde{a}_{2,1}}.
\]
By substituting these expressions into $q$ we obtain two univariate polynomials
of degree 10 in the variable $\rho_2$:
\begin{eqnarray*}
\mathfrak{v}_1 &=& q_{2,0}\tilde{a}^2_{1,0} - q_{1,0}\tilde{a}_{1,0}\tilde{a}_{1,1} + 
b_0\tilde{a}_{1,1}^2 , \\
\mathfrak{v}_2 &=& q_{2,0}\tilde{a}^2_{2,0} - q_{1,0}\tilde{a}_{2,0}\tilde{a}_{2,1} + 
b_0\tilde{a}_{2,1}^2. 
\end{eqnarray*}
Using the properties of resultants we find that
\[
\mathfrak{u}_1 = q_{2,0}^3 \mathfrak{v}_1, \qquad \mathfrak{u}_2 = q_{2,0}^4 \mathfrak{v}_2.
\]

\subsection{Non-degeneracy conditions}

We list below the conditions on the data $\Att_1, \Att_2,
\qu_1,\qu_2,\qudot_1,\qudot_2$ that we used in the previous sections. These
conditions generically hold.

\begin{enumerate}
  
\item $\EE_1\cdot\DD_1\times\DD_2, \EE_2\cdot\DD_1\times\DD_2 \neq 0$,
  so that $q$ is a quadratic polynomial both in $\rho_1$ and $\rho_2$. An
  interpretation of these relations is given in
  \cite{gdm10}. Moreover, these conditions imply:
  \begin{itemize} 
  \item[i)] $\DD_1\times\DD_2\neq \bzero$, so that we can compute
    $\rhodot_j=\rhodot_j(\rho_1,\rho_2)$ from system
    (\ref{eq_AM}). This condition also implies $\DD_1,\DD_2\neq\bzero$,
    which in turn yield $\erre_1,\erre_2\neq\bzero$ for all
    $\rho_1,\rho_2$;
    
  \item[ii)] $\erho_1\times\eort_1\cdot\qu_1,
    \erho_2\times\eort_2\cdot\qu_2 \neq 0$, which are used to define
    $\rho_1'', \rho_2''$ respectively;
    
  \item[iii)] $\erho_1\times\erho_2\cdot\qu_1,
    \erho_1\times\erho_2\cdot\qu_2\neq 0$, which are used to define $\rho_2',
    \rho_1'$ respectively.

  \end{itemize}
  
\item $\qu_1\times\qu_2\neq\bzero$, $\rho_j'\neq\rho_j''$, $\rho_j''\neq 0$,
  $j=1,2$.
  These conditions imply $\erre_1\times\erre_2\neq\bzero$, and also
  $\erre_1,\erre_2,\bDelta_r\neq \bzero$, for all $\rho_1,\rho_2$.

\item $\erho_j\times\eort_j\cdot\qudot_j\neq 0$, so that
  $\erredot_j\neq\bzero$ for all $\rho_j,\rhodot_j$, with $j=1,2$.



\item $\bDelta_q\cdot\erho_1\times\erho_2 \neq 0$, so that
  $\{\erho_1,\erho_2,\bDelta_r\}$ forms a basis of $\R^3$. This condition
  implies $\erho_1\times\erho_2\neq\bzero$, so that the maximal total degree
  for the components of equation (\ref{eqdeg6}) is 6.

\item $p_1(P_2), p_2(P_1)\neq 0$, so that $P_1$, $P_2$ are not solutions of
  (\ref{qxi}).

\end{enumerate}

\section{Selecting the solutions}
\label{s:selsol}

After computing all the solutions of (\ref{qxi}) we can select the ones with
both entries $\rho_1, \rho_2$ real and positive and compute the corresponding
values of $\rhodot_1, \rhodot_2$.  However, since equations (\ref{sys2solve})
impose only some of the laws of the two-body dynamics, we expect that some of
the remaining solutions yield vectors
$(\rho_1,\rho_2,\rhodot_1,\rhodot_2,u_1,u_2)$ which do not solve
(\ref{fullsys}).

\noindent In this section we characterize the spurious solutions, and
propose some algorithms to select the good ones.
To decide which solutions can be accepted we suggest to use one of the two
methods introduced in \cite{gdm10}, \cite{gfd11}.
They take into account the errors in the observations, which can be
represented by $4\times 4$ covariance matrices $\Gamma_1$, $\Gamma_2$ of the
attributables.
We recall that the first method relies on the computation of a norm
referring to some compatibility conditions, see \cite{gdm10}, Section
5.
The second method requires to compute a covariance matrix for the candidate
preliminary orbits, which is used for the attribution algorithm, see
\cite{gfd11}, Sections 7, 8.  The new formulas for the covariance matrix
are provided in Section~\ref{s:covar}.

\subsection{Spurious solutions}
\label{s:spurious}

We consider real solutions of (\ref{sys2solve}) that do not
solve the intermediate system (\ref{intersys}), and solutions of
(\ref{intersys}) that do not solve (\ref{fullsys}).
In the first case the spurious solutions do not satisfy
\[
\bDelta_K\cdot\bDelta_r + \Bigl(\frac{1}{2}|\erredot_1|^2 -
\frac{\mu}{|\erre_1|}\Bigr)|\bDelta_r|^2 = 0.
\]
Concerning the second case we prove the following:
\begin{proposition}
Each real solution of (\ref{intersys}) which does not solve
(\ref{fullsys}) fulfills
\[
|\lenz_1 - \lenz_2|= 2.
\]
\label{spurious2}
\end{proposition}

\begin{proof}
\noindent Using $\angmom_1 =\angmom_2$, the second equation in
(\ref{intersys}) can be written as
\begin{equation}
(\lenz_1 - \lenz_2) = \mu\frac{(|\lenz_1|^2 -
  |\lenz_2|^2)}{2|\angmom|^2}\erre_2,
\label{eqdifflenz}
\end{equation}
where $\angmom$ is the common value of the angular momentum. The
conservation of Laplace-Lenz vector and energy is in general not
guaranteed.  However, we note that if $\angmom_1=\angmom_2$ and
$|\lenz_1| =|\lenz_2|$, then (\ref{intersys}), (\ref{eqdifflenz})
imply $\lenz_1=\lenz_2$ and $\energy_1=\energy_2$.

\noindent If $|\lenz_1| >|\lenz_2|$, the vectors $\lenz_1-\lenz_2$,
$\erre_2$ have the same orientation.
Passing to the norms in (\ref{eqdifflenz}) and substituting
\[
|\erre_2|=\frac{|\angmom|^2}{\mu(1+|\lenz_2|\cos\theta_2)},
\]
where $\theta_2$ is the angle between $\lenz_2$ and $\erre_2$, we obtain
\[
|\lenz_1 - \lenz_2|=\frac{|\lenz_1|^2 -|\lenz_2|^2}{2(1+|\lenz_2|\cos\theta_2)}.
\] 
Using relation
\[
|\lenz_1|^2 -|\lenz_2|^2=|\lenz_1 - \lenz_2|^2+2|\lenz_1 -
\lenz_2||\lenz_2|\cos\theta_2
\]  
and rearranging the terms we get $|\lenz_1 - \lenz_2|=2$.

\noindent If $|\lenz_1| <|\lenz_2|$ the vectors $\lenz_1-\lenz_2$,
$\erre_2$ have opposite orientation, and we obtain $|\lenz_1 -
\lenz_2|=-2$, which is impossible.
In fact in this case we have (see Figure~\ref{fig:2cases})
\[
|\lenz_1|^2 -|\lenz_2|^2=|\lenz_1 - \lenz_2|^2+2|\lenz_1 -
\lenz_2||\lenz_2|\cos(\pi-\theta_2).
\]
\begin{figure}[t]
\centerline{\includegraphics[width=5.5cm]{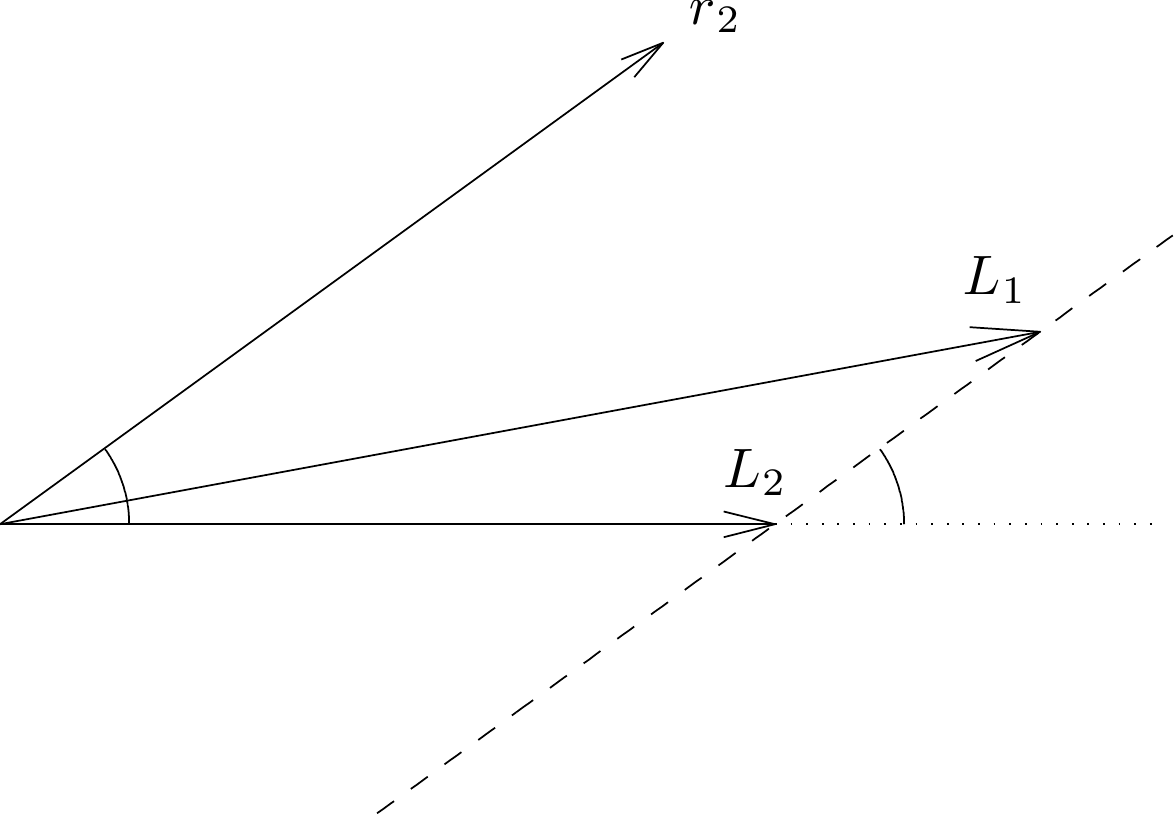}
\hskip 0.8cm\includegraphics[width=5.5cm]{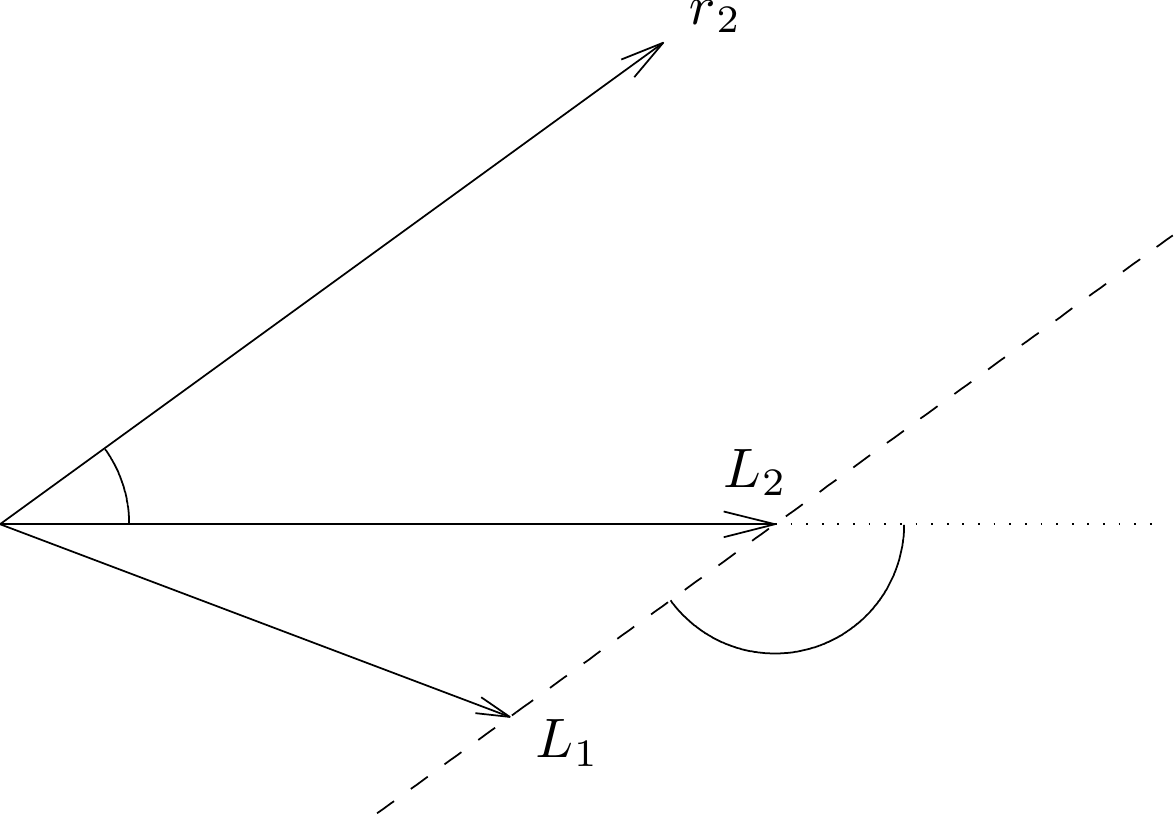}}
\caption{Left: case $|\lenz_1|>|\lenz_2|$.  Right: case $|\lenz_1|<|\lenz_2|$.}
\label{fig:2cases}
\end{figure}
\rightline{$\square$\hskip 0.5cm}

\end{proof}

\noindent Since system (\ref{intersys}) is generically inconsistent, the
spurious solutions discussed in Proposition~\ref{spurious2} usually do not
occur.

\subsection{Covariance of the solutions}
\label{s:covar}


We introduce the vectors
\[
\Avec = (\Att_1,\Att_2),
\qquad 
\Rvec
= (\rho_1, \rhodot_1, \rho_2, \rhodot_2)
\]
and the map
\[
\bPhi(\Rvec, \Avec) = \left(
\begin{array}{c}
\angmom_1 -\angmom_2
\cr
\bDelta_K \times \bDelta_r \cdot\erho_1\cr
\end{array}
\right) .
\]

\noindent Assume we have a pair of attributables $\bar{\Avec} = 
(\bar{\Att}_1,\bar{\Att}_2)$, with
covariance matrices $\Gamma_1, \Gamma_2$, at epochs $\bar{t}_1$, $\bar{t}_2$.
For each vector $\bar{\Rvec}$ such that
\[
\bPhi(\bar{\Rvec}, \bar{\Avec})=\bzero,
\hskip 1cm
\det\frac{\partial \bPhi}{\partial \Rvec}(\bar{\Rvec},\bar{\Avec})\neq 0
\]
there exists a map
$\Avec\mapsto\Rvec(\Avec)$ with $\bPhi(\Rvec(\Avec), \Avec) = \bzero$,
defined in a neighborhood of $\bar{\Avec}$, with $\Rvec(\bar{\Avec})=\bar{\Rvec}$.
Following \cite{gfd11}, we consider the map $\bPsi$ defined by
$\bPhi=\bPsi\circ\mathcal{T}_{att}^{car}$, where $\mathcal{T}_{att}^{car}$ is
the transformation from attributable coordinates  ${\bf E}_{att}$ to Cartesian
coordinates ${\bf E}_{car}$ at the two epochs.
Then we can use the same scheme as in \cite{gfd11} (Section 7) to compute the 
covariance matrix of the Cartesian coordinates at epoch $\bar{t}_1$
through the formula
\[
  \Gamma_{car}^{(1)} = \frac{\partial E_{car}^{(1)}}{\partial\Avec}
\Gamma_{\Avec}\Bigl[\frac{\partial
  E_{car}^{(1)}}{\partial\Avec}\Bigr]^T,
\hskip 1cm
\Gamma_{\Avec}= \left[
\begin{array}{cc}
\Gamma_1 &0\cr
0 &\Gamma_2\cr
\end{array}
\right],
\]
with the derivatives $\frac{\partial E_{car}^{(1)}}{\partial\Avec}$ evaluated
at $\Avec=\bar{\Avec}$.

\noindent The only differences with respect to \cite{gfd11} are in the term
$\frac{\partial\bPsi}{\partial{\bf E}_{car}}$, which in this case is given by
\[
\frac{\partial\bPsi}{\partial\mathbf{E}_{car}} = 
\left[
\begin{array}{cccc}
-\widehat{{\erredot}_1}  &\widehat{{\erre}_1}  &\widehat{{\erredot}_2}  &-\widehat{{\erre}_2}\cr
\stackrel{}{\displaystyle\frac{\partial p_1}{\partial \erre_1}} 
&\displaystyle\frac{\partial p_1}{\partial \erredot_1} 
&\displaystyle\frac{\partial p_1}{\partial \erre_2} 
&\displaystyle\frac{\partial p_1}{\partial \erredot_2}\cr
\end{array}
\right],
\]
where we use 
the {\em hat map}
\[
\R^3\ni (u_1,u_2,u_3) = {\bf u} \mapsto \widehat{\bf u} =
\left[
\begin{array}{ccc}
0 &-u_3 &u_2\cr
u_3 &0 &-u_1\cr
-u_2 &u_1 &0\cr
\end{array}
\right]  .
\]
To compute the derivatives of $p_1=\bDelta_K \times \bDelta_r \cdot\erho_1$ we
take advantage of the following relation, which holds at the solutions of $p_1=0$:
\[
\frac{\partial p_1}{\partial\mathbf{E}_{car}} = \frac{1}{\rho_1}\frac{\partial
  p_1^*}{\partial\mathbf{E}_{car}}.
\]
Here $p_1^* = \bDelta_K\times\bDelta_r\cdot \bv_1$, with
$\bv_1=\erre_1-\qu_1$, and
\begin{eqnarray*}
\frac{\partial p_1^*}{\partial \erre_1}  &=& 
\frac{1}{2}(|\erredot_2|^2-|\erredot_1|^2)\qu_1\times\erre_2 -
(\erredot_1\times\bDelta_r\cdot\bv_1) \erredot_1 \\
&& -(\erredot_1\cdot\erre_1)[-\qu_1\times\erredot_1 + \erre_2\times\erredot_1] +
(\erredot_2\cdot\erre_2)[-\qu_1\times\erredot_2 + \erre_2\times\erredot_2],
\\
\frac{\partial p_1^*}{\partial \erredot_1} &=&
- (\erre_1\times\erre_2\cdot\bv_1) \erredot_1 -
(\erredot_1\times\bDelta_r\cdot\bv_1) \erre_1 -
(\erredot_1\cdot\erre_1)\bDelta_r\times\bv_1,
\\
\frac{\partial p_1^*}{\partial \erre_2} &=& 
-\frac{1}{2}(|\erredot_2|^2-|\erredot_1|^2)\qu_1\times\erre_1 
- (\erredot_1\cdot\erre_1)\erredot_1\times\bv_1 \\
&& +(\erredot_2\times\bDelta_r\cdot\bv_1) \erredot_2
+
(\erredot_2\cdot\erre_2)\erredot_2\times\bv_1,
\\
\frac{\partial p_1^*}{\partial \erredot_2} &=&
(\erre_1\times\erre_2\cdot\bv_1) \erredot_2 +
(\erredot_2\times\bDelta_r\cdot\bv_1) \erre_2 +
(\erredot_2\cdot\erre_2)\bDelta_r\times\bv_1.
\end{eqnarray*}

\section{Including the $J_2$ effect}
\label{s:J2}

If we want to compute preliminary orbits of space debris, we have to take into
account the effect of the oblateness of the Earth ($J_2$ effect), which gives
rise to a perturbation of the two-body dynamics.  Following \cite{ftmr10} and
\cite{gfd11} we can use an iterative scheme to include the $J_2$ effect in the
determination of preliminary orbits with the two-body integrals.
We consider the equations
\[
R_c\angmom_1 = \angmom_2, \quad
 R_L\lenz_1 = \lenz_2,\quad
\energy_1 = \energy_2,\quad
u_1^2|\erre_1|^2 = \mu^2, \quad
u_2^2|\erre_2|^2 = \mu^2,
\]
where the rotation matrices 
\begin{equation}
R_c = R_{\Delta\Omega}^{\hat{z}},
\qquad
R_L = R_{\omega_1 + \Delta\omega}^{\angmom_2} R_{\Delta\Omega}^{\hat{z}}
R_{-\omega_1}^{\angmom_1}
\label{rotmat}
\end{equation}
are defined through the angles
\[
\Delta\omega = \omega_2-\omega_1,
\hskip 0.5cm \Delta\Omega = \Omega_2-\Omega_1 .
\]
In (\ref{rotmat}) we use $R_\theta^\bv$ to denote the rotation by
the angle $\theta$ around the axis defined by the vector $\bv$.

\noindent Following the same steps of Section~\ref{s:elim}
we consider the intermediate equation
\[
\mu R_L\lenz_1 - \energy_1\erre_2 = \mu\lenz_2 - \energy_2\erre_2,
\]
which can be written
\[
\bDelta_K^L + (\frac{1}{2}|\erredot_1|^2 - u_1)\bDelta_r^L = \bzero,
\]
where
\[
\bDelta_K^L = R_L\bK_1 - \bK_2, \hskip 1cm
\bDelta_r^L = R_L\erre_1 - \erre_2.
\]
Then we can eliminate the dependence on $u_1$ by vector product with
$\bDelta_r^L$.
We end up with the system
\begin{equation}
R_c\angmom_1 = \angmom_2,
\qquad
\bDelta_K^L\times \bDelta_r^L = \bzero.
\label{sys2solveJ2}
\end{equation}
Note that system (\ref{sys2solveJ2}) is not polynomial due to the presence of
the rotation matrices $R_c, R_L$, that depend on the orbital elements.
%
%
We can search for solutions of (\ref{sys2solveJ2})
by considering the solutions of
(\ref{sys2solve}), i.e. assuming $R_L=R_c=I$, as first guesses of the
iterative method.  Inserting these solutions into 
$R_L$ and $R_c$, system (\ref{sys2solveJ2}) becomes polynomial, and we can solve
it like in the unperturbed case.
We iterate the procedure and consider the solutions at convergence, if any.
Some care must be taken in selecting the solutions at each iteration.

Similarly to (\ref{DeltaKxDeltar}) we have
\[
\bDelta_K^L\times \bDelta_r^L =
\frac{1}{2}(|\erredot_2|^2- |\erredot_1|^2)R_L\erre_1\times\erre_2
- (\erredot_1\cdot\erre_1)R_L\erredot_1\times\bDelta_r^L +
(\erredot_2\cdot\erre_2)\erredot_2\times\bDelta_r^L.
\]
By developing the expressions of $R_L\erre_1\times\erre_2$,
$R_L\erredot_1\times\bDelta_r^L$, $\erredot_2\times\bDelta_r^L$ as
polynomials in $\rho_1,\rho_2$ we find that the monomials in
$\bDelta_K^L\times\bDelta_r^L$ with the highest degree (i.e. 6) are
all multiplied by $R_L\erho_1\times\erho_2$. Then we can project the
second equation in (\ref{sys2solveJ2}) onto the vectors
$R_L\erho_1$, $\erho_2$,
to obtain two polynomial equations
$p_j(\rho_1,\rho_2)=0$, $j=1,2$.  As before, $p_1$ and $p_2$ have degree $5$.
\noindent The range rates can be eliminated from equation $R_c\angmom_1 =
\angmom_2$, to get a quadratic polynomial $q$ analogous to (\ref{quad_form}).

An important difference with respect to Section~\ref{s:elim} is that in this
case the system $p_1=p_2=q=0$ is generically inconsistent, as can be checked
by a numerical test. Nevertheless we can choose either $p_1=q=0$ or $p_2=q=0$
as polynomial equations for the linkage.

\section{Numerical tests}
\label{s:numexp}

\begin{figure}[!h]
\centerline{
\includegraphics[width=7cm]{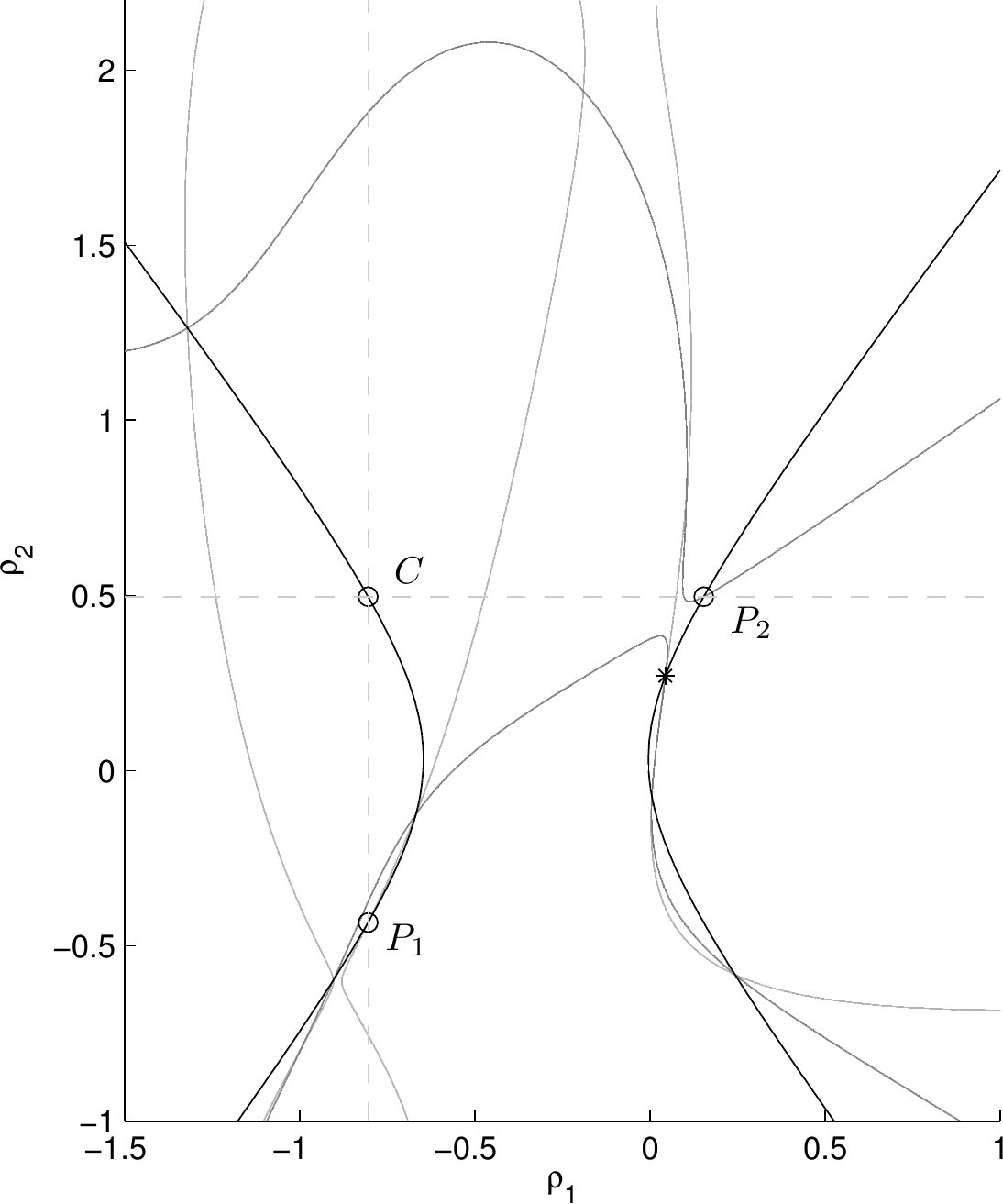}
\includegraphics[width=5.5cm]{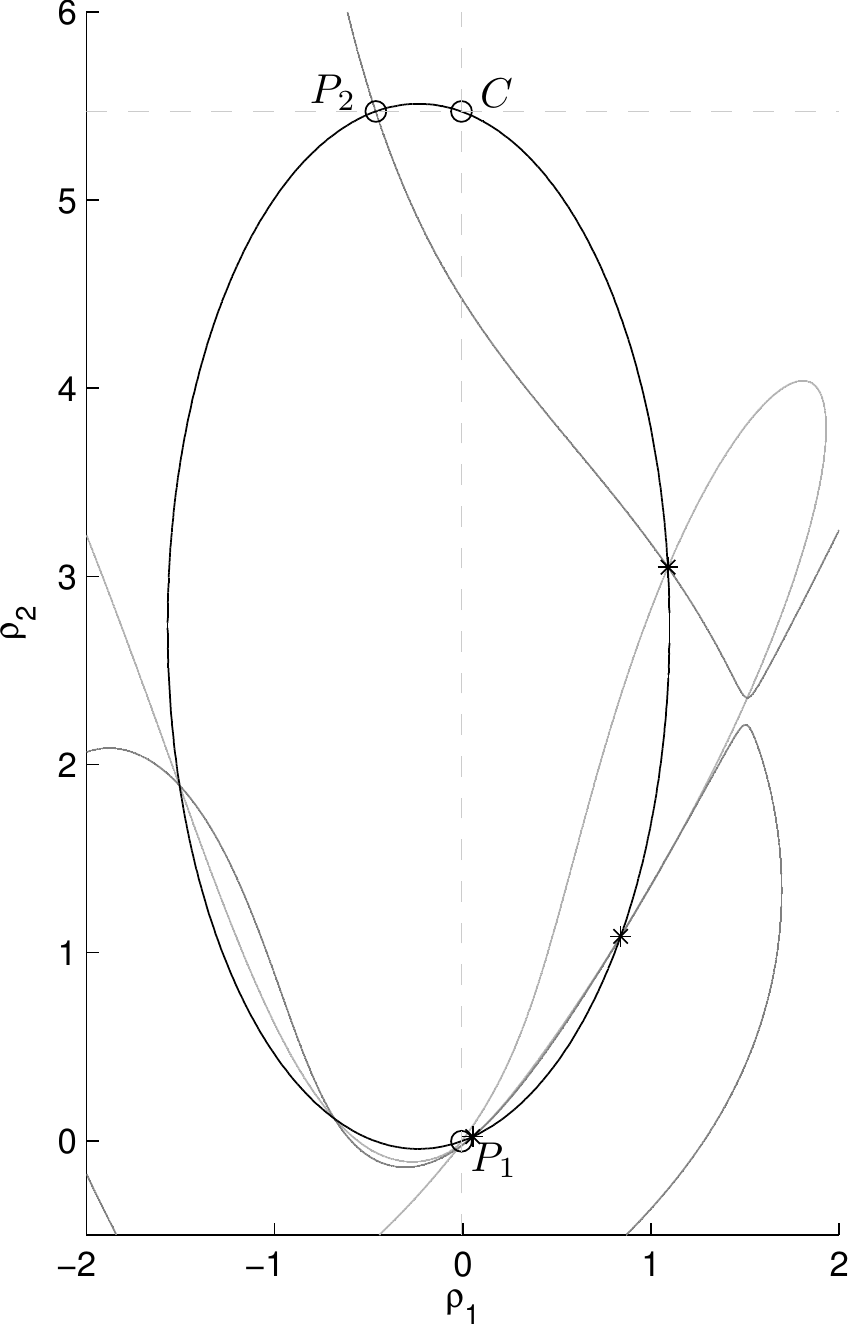}
}
\caption{Left: Intersection of curves defined by $q=p_1=p_2=0$ for
  (101955) Bennu.  Right: the same for (100000) Astronautica.} 
\label{fig:examples}
\end{figure}

We describe the results of two numerical tests, one for a near-Earth
asteroid, the other for a main belt.
The first object is (101955) Bennu.  We link 5 observations 
made on September 11, 1999 together with 11 observations
made on March 30, 2000. After discarding non-real and non-positive
solutions, we are left with only one pair
\[
(\rho_1,\rho_2) = (0.04379, 0.27132).
\]
The second object is the main belt asteroid (100000) Astronautica. We
link 10 observations made on October 17 and 19, 2003 together with 4
observations made on May 9, 2005.  In this case we are left with the
pairs
\[
(\rho_1,\rho_2) = (0.05240, 0.02179), \quad (0.83897, 1.08648), \quad
(1.09052, 3.04874).
\]
The third pair is discarded because it yields an unbounded orbit at epoch
$\bar{t}_2$.  The values of the penalty (see \cite{gfd11}) for the first and second
solution, computed by attribution, are $\chi_4= 1224012.479$ and
$\chi_4= 1.497$ respectively, therefore we select the second solution.
 
\noindent In Figure~\ref{fig:examples} we draw the curves $q=0$ (black),
$p_1=0$ (light gray), $p_2=0$ (dark gray). The dashed straight lines
correspond to $c_{11}=0$ (vertical), $c_{22}=0$ (horizontal).  The
computed pairs $(\rho_1,\rho_2)$ are marked with asterisks.

\noindent In Table~\ref{tab:elems} we show for both asteroids the Keplerian
elements at the two mean epochs $\bar{t}_1, \bar{t}_2$ that we find with this
method. For asteroid (100000), we indicate with the labels 1, 2 the two
bounded orbits that we can compute.
\begin{table}
\begin{tabular}{l|c|c|c|c|c|c|c}
asteroid &$MJD$ &$a$ &$e$ &$I$ &$\Omega$ &$\omega$ &$\ell$ \cr
\hline
$(101955)$   &51432.4 &1.1316    &0.2058    &5.9895    &2.1229   &65.2425  &303.8601  \cr
             &51633.4 &1.1306    &0.2039    &5.9895    &2.1229   &66.4329  &107.6598  \cr
\hline
$(100000)_1$ &52930.2 &1.0365    &0.0277   &1.0298  &202.3124  &275.0221  &268.4023 \cr
             &53499.3 &1.0367    &0.0305   &1.0298  &202.3124  &308.8393   &73.7641  \cr
\hline
$(100000)_2$ &52930.2 &1.8725    &0.0768   &20.8665  &186.7970  &197.8094  &344.8383  \cr
             &53499.3 &1.8754    &0.0862   &20.8665  &186.7970  &197.7172  &198.9997  \cr
\end{tabular}
\caption{Keplerian elements found for (101955) Bennu, and for
  (100000) Astronautica.}
\label{tab:elems}
\end{table}

\section{Acknowledgments}
The work is partially supported by the Marie Curie Initial Training
Network Stardust, FP7-PEOPLE-2012-ITN, Grant Agreement 317185.

\appendix
\section{Appendix}
\label{s:app}

We suggest a way to discard pairs of attributables which are not likely to
fulfill the requirements to be linked. This filter is based on the angular
momentum equation $\angmom_1=\angmom_2$.

\subsection{Filtering pairs of attributables} 
\label{s:filter}

The equation $q=0$ represents a conic $Q$ in the plane $\rho_1\rho_2$.
We can decide to accept only pairs $(\rho_1,\rho_2)$ inside a square
${\cal R}=[\rho_{min},\rho_{max}]\times[\rho_{min},\rho_{max}]$, with
$0<\rho_{min}<\rho_{max}$.  We check whether the conic $Q$ does
intersect ${\cal R}$.

\noindent First consider the case $Q$ is unbounded (hyperbola or
parabola).  In this case it is sufficient to check whether the conic
intersects the boundary $\partial\mathcal{R}$.

\noindent If $Q$ is bounded (ellipse or circle) this check is not
enough: it can happen that $Q\cap \mathcal{R}\neq\emptyset$, with $Q$ lying
totally inside $\mathcal{R}$.

\noindent Here we sketch an algorithm for this filter.
Consider the four straight lines
\[
r_j = \{(\rho_1,\rho_2) : \rho_j = \rho_{min}\}, \quad r_{j+2} =
\{(\rho_1,\rho_2) : \rho_j = \rho_{max}\}, \qquad\mbox{ for } j=1,2.
\]
First compute the intersections of the conic with each line $r_j$, $j=1\ldots
4$, if any.  If an intersection lies on the segment of some $r_j$ belonging
to $\partial{\cal R}$ we accept the pair of attributables and continue the
linkage procedure. If not, we check whether $Q$ is bounded or unbounded.

\noindent If $Q$ is unbounded we discard the pairs of attributables.  If $Q$
is bounded, we can check whether $Q$ is totally inside ${\cal R}$ by computing
the coordinates of the center of the conic.  If the center is inside ${\cal
  R}$ we accept the pair provided $Q$ has no intersections with any $r_j$,
otherwise we reject it.



\end{document}